\documentclass[aip,jmp,amsmath,amssymb,preprint,onecolumn]{revtex4-1}
\usepackage[T1]{fontenc}
\usepackage{amsthm}
\usepackage{bm}
\usepackage{mathrsfs}
\usepackage{graphicx}
\usepackage{xcolor}
\colorlet{BLUE}{blue}
\usepackage{hyperref}
\hypersetup{hidelinks}
\newtheorem{theorem}{Theorem}
\newtheorem{lemma}[theorem]{Lemma}
\newtheorem{corollary}[theorem]{Corollary}
\newtheorem{proposition}[theorem]{Proposition}
\theoremstyle{definition}
\newtheorem{definition}[theorem]{Definition}
\newtheorem{remark}[theorem]{Remark}

\newcommand{\dd}{\,\mathrm{d}}
\newcommand{\ii}{\mathrm{i}}

\newcommand{\R}{\mathbb{R}}
\newcommand{\C}{\mathbb{C}}
\newcommand{\N}{\mathbb{N}}

\newcommand{\Nzero}{\mathbb{N}_{0}}
\newcommand{\HK}{\mathcal{H}_{\mathrm K}}
\newcommand{\Dom}{\operatorname{Dom}}

\newcommand{\ACloc}{\operatorname{AC}_{\mathrm{loc}}}
\newcommand{\Span}{\operatorname{span}}

\newcommand{\Ree}{\operatorname{Re}}

\newcommand{\vh}{\hat v}

\newif\ifreviewchanges
\reviewchangesfalse

\newcommand{\rev}[1]{{\ifreviewchanges\color{blue}\fi#1}}
\begin{document}
\title{Laplace--King representations: density and spectral theory}

\author{Yanpeng Wang}
\affiliation{Department of Engineering Physics, Tsinghua University, Beijing 100084, China}

\author{Zhe Gao}
\email{gaozhe@tsinghua.edu.cn}
\affiliation{Department of Engineering Physics, Tsinghua University, Beijing 100084, China}
\affiliation{Beijing Primal Energy Theory Research Institute, Beijing 100084, China}

\begin{abstract}

Laplace--King representations combine spherical harmonics with
King functions [Wang et al., Chin. Phys. B \textbf{34}, 065201
(2025)], the radial kernels of shifted isotropic Gaussians.
The radial parameters may vary across angular modes.
We prove that, for every angular degree, fixed-width kernels
with positive real shifts have dense complex linear span in a
Gaussian-weighted radial \(L^2\) space.
Finite Laplace--King representations are dense in the
corresponding three-dimensional weighted space; allowing
variable widths preserves density in the same reference norm.
A generating identity connects the kernels to generalized
Laguerre polynomials.
The self-adjoint King operator is unitarily equivalent to the
free radial Schr\"odinger operator; its spectral resolution
defines a continuous King mixture model (KMM) through
imaginary-shift kernels in a distinct weighted Hilbert space.

\end{abstract}

\maketitle

% ============================================================
\section{Introduction and Main Results}
\label{sec:Introduction}

Probability-density models are used in statistical inference
\cite{McLachlan2000}, image classification
\cite{Sanchez2013Fisher}, and deep generative clustering
\cite{Jiang2017VaDE}.
Gaussian mixture models (GMMs) provide tractable representations
in these settings and have also been used in numerical solutions
of the Fokker--Planck equation
\cite{Tabandeh2022PhysicsBasedMixtures}.
On \(\R^d\), a finite full-covariance GMM has the form

\begin{equation}
 p(\bm x)=\sum_{s=1}^{S}\pi_s
 \mathcal N_d(\bm x;\bm\mu_s,\bm\Sigma_s),
 \qquad
 \pi_s\ge0,\quad\sum_{s=1}^{S}\pi_s=1,
 \label{eq:intro-GMM}
\end{equation}
\noindent where
\begin{equation}
 \mathcal N_d(\bm x;\bm\mu_s,\bm\Sigma_s)
 =\frac{\exp\!\left[-\frac12
    (\bm x-\bm\mu_s)^{\mathsf T}\bm\Sigma_s^{-1}
    (\bm x-\bm\mu_s)\right]}
 {(2\pi)^{d/2}(\det\bm\Sigma_s)^{1/2}}.
 \label{eq:intro-Gaussian}
\end{equation}
\noindent The component covariances \(\bm\Sigma_s\) are arbitrary
symmetric positive-definite matrices.
\noindent
Finite positive Gaussian location--scale mixtures are dense in
\(L^1(\R^d)\) among probability densities
\cite{Nguyen2023LocationScale}.

For three-dimensional distributions, spherical harmonics
separate angular dependence into modes with radial amplitudes,
providing a natural way to exploit symmetry
\cite{Tzoufras2011OSHUN}.
We previously introduced King functions
\cite{wang2024Aconservative,wang2025general} as the radial kernels
of shifted isotropic Gaussians.
For speed \(v\), shift magnitude \(u\), and width
\(\varsigma>0\), they take the form
\begin{align}
 \mathcal K_l^{\mathrm{phys}}(v;u,\varsigma)
 &=\frac{2l+1}{\pi^{3/2}\varsigma^3}
   e^{-(v^2+u^2)/\varsigma^2}
   i_l\!\left(\frac{2uv}{\varsigma^2}\right),
 \label{eq:Kl}
\end{align}
\noindent
where \(i_l (\cdot)\) is the modified spherical Bessel function of the
first kind.
The King mixture model (KMM) uses these kernels to approximate
the radial amplitudes of nonequilibrium velocity distributions\cite{wang2024Aconservative,wang2025general}.

In this work, we further extend this construction to the Laplace--King representation.
In a GMM with isotropic components, one set of component
weights, drift vectors, and widths determines all
spherical-harmonic radial amplitudes.
With the full angular expansion retained, such a GMM is a
shared-parameter special case of the
\hyperref[eq:f-Laplace-King]{Laplace--King representation}.
The general form allows the radial parameters and component
counts to be chosen independently for each angular mode \((l,m)\).

This mode-dependent construction raises the question of whether
finite KMMs can approximate arbitrary radial amplitudes in an
appropriate function space.
We establish this approximation property and its
three-dimensional counterpart, and develop the spectral theory
of the fixed-width King operator.

\rev{We use a fixed reference scale \(\varsigma>0\) and radial
coordinate \(\vh=v/\varsigma\).
For each angular degree, the fixed-width King kernels with
positive real shifts have dense complex linear span in the
Gaussian-weighted radial space \hyperref[eq:Hrad]{\(H_{\mathrm{rad}}\)}
(Theorem~\ref{thm:King-density}).
Together with the angular Parseval identity, this establishes
density of finite Laplace--King representations in the
corresponding three-dimensional weighted space
\(\mathscr H_{\mathrm{3V}}\)
(Theorem~\ref{thm:direct-KMM-density}).
Every finite full-covariance GMM is therefore an admissible
approximation target.
For a countable dictionary of distinct positive shifts with a
finite accumulation point,
\hyperref[cor:countable-KMM-projections]{orthogonal projections}
onto successive finite spans converge in \(H_{\mathrm{rad}}\).
The \hyperref[cor:variable-width-density]{variable-width class}
contains the fixed-width subfamily and is dense in the same
reference norm.}

\rev{A \hyperref[thm:main-King-Laguerre]{King--Laguerre generating identity}
connects the fixed-width kernels to generalized Laguerre polynomials.
The associated King operator is unitarily equivalent to the
free radial Schr\"odinger operator; its self-adjoint spectral
resolution defines the \hyperref[subsec:continuous-KMM]{continuous KMM}
(Theorems~\ref{thm:main-unitary} and~\ref{thm:main-spectral}).
Here imaginary shifts parametrize the generalized eigenfunctions.
The spectral representation uses a different weighted Hilbert space,
\hyperref[eq:HK]{\(\HK\)}, from the approximation space
\(H_{\mathrm{rad}}\).}

\rev{For finite representations, an
\hyperref[eq:KMM-total-error-budget]{exact error identity}
separates angular truncation from radial approximation.
Explicit \hyperref[prop:moment]{variable-width moments}
give linear constraints on the coefficients for normalization
and other prescribed moments when the kernel parameters are fixed.}

Section~\ref{sec:Preliminaries} introduces the King kernels and
defines the \hyperref[subsec:Laplace-King]{finite representations}.
Section~\ref{sec:Density} establishes their density and
\hyperref[subsec:consequences-density]{projection convergence}.
Section~\ref{sec:LaguerreStructure} develops the
\hyperref[subsec:proof-King-Laguerre]{King--Laguerre generating identity}
and the \hyperref[subsec:renormalized-Laguerre]{weighted-space correspondence}.
Section~\ref{sec:KingOperator} gives the
\hyperref[subsec:unitary]{self-adjoint realization} of the King
operator and its \hyperref[subsec:spectrum]{spectral resolution}.
Section~\ref{sec:Applications} treats
\hyperref[subsec:GMM-subclass]{probability constraints},
\hyperref[subsec:approximation-budget]{approximation errors}, and
\hyperref[subsec:moment-evaluation]{moment evaluation}, and outlines
a \hyperref[subsec:kinetic-example]{prospective Fokker--Planck application}.
\rev{Section~\ref{sec:Conclusion} summarizes the results and
directions for further work.}

\section{King kernels and finite representations}
\label{sec:Preliminaries}
\label{sec:King}

\subsection{Angular normalization}
\label{subsec:Gaussian-angular}

We use the unnormalized spherical harmonics
\begin{equation}
 \widetilde Y_l^m(\theta,\phi)
 =P_l^m(\cos\theta)e^{\ii m\phi},
 \qquad l\in\Nzero,\quad -l\le m\le l,
 \label{eq:Ylm}
\end{equation}
\noindent where \(P_l^m\) includes the Condon--Shortley phase, with
\(P_l^{-m}=(-1)^m(l-m)!P_l^m/(l+m)!\) for \(0\le m\le l\).
Writing
\[
 N_l^m=\sqrt{\frac{2l+1}{4\pi}\frac{(l-m)!}{(l+m)!}},
\]
\noindent the harmonics \(N_l^m\widetilde Y_l^m\) are orthonormal.
For \(f(v,\cdot)\in L^2(\mathbb S^2)\) at almost every \(v>0\),
the Laplace series
\begin{equation}
 f(v,\Omega)
 =\sum_{l=0}^{\infty}\sum_{m=-l}^{l}
   f_l^m(v)\widetilde Y_l^m(\Omega)
 \label{eq:f-spherical-expansion}
\end{equation}
\noindent converges in \(L^2(\mathbb S^2)\) at almost every such \(v\)
\cite{Arfken2013, Freeden1998}, with coefficients
\begin{equation}
 f_l^m(v)
 =(N_l^m)^2\int_{\mathbb S^2}
 f(v,\Omega)\overline{\widetilde Y_l^m(\Omega)}\,\dd\Omega,
 \qquad \dd\Omega=\sin\theta\,\dd\theta\,\dd\phi.
 \label{eq:f-radial-coefficient}
\end{equation}
\noindent The same angular convention applies at fixed \(\vh\).
Modal amplitudes can be signed or complex even for a nonnegative
density; they must not be identified with probability weights.

\subsection{Gaussian origin of the radial kernel}
\label{subsec:KingFun}

For \(\bm u\in\R^3\) and \(\varsigma>0\), consider
\begin{equation}
 M_{\bm u,\varsigma}(\bm v)
 =\frac{1}{\pi^{3/2}\varsigma^3}
   \exp\!\left(-\frac{|\bm v-\bm u|^2}{\varsigma^2}\right).
 \label{eq:fDM}
\end{equation}

\noindent Its covariance is \((\varsigma^2/2)I_3\).
Write \(v=|\bm v|\) and \(u=|\bm u|\), with directions
\(\Omega,\Omega_u\) for nonzero vectors.
Using the angular convention \eqref{eq:Ylm}, define
\begin{align}
 B_l^m(\Omega_u)
 &:=\frac{(l-m)!}{(l+m)!}
      \overline{\widetilde Y_l^m(\Omega_u)}.
 \label{eq:angular-direction-factor}
\end{align}

\begin{proposition}[King function from a shifted Gaussian]
\label{prop:King-kernel-from-Gaussian}
For \(u>0\), the Gaussian \eqref{eq:fDM} has the exact
spherical-harmonic expansion
\begin{equation}
 M_{\bm u,\varsigma}(v,\Omega)
 =\sum_{l=0}^{\infty}\sum_{m=-l}^{l}
   \mathcal K_l^{\mathrm{phys}}(v;u,\varsigma)\,
   B_l^m(\Omega_u)\,\widetilde Y_l^m(\Omega),
 \label{eq:fDM-Ylm-general}
\end{equation}
\noindent with the King kernel given by \eqref{eq:Kl}.
The case \(u=0\) is its centered limit.
\end{proposition}

\begin{proof}

For \(u,v>0\), set \(\bm e_{\bm v}=\bm v/v\) and
\(\bm e_{\bm u}=\bm u/u\).  Factor the Gaussian as
\[
 M_{\bm u,\varsigma}(v,\Omega)
 =\frac{e^{-(v^2+u^2)/\varsigma^2}}{\pi^{3/2}\varsigma^3}
  e^{z\bm e_{\bm v}\cdot\bm e_{\bm u}},
 \qquad z=\frac{2uv}{\varsigma^2}.
\]
\noindent The plane-wave expansion \cite{Arfken2013, NIST2010} is
\begin{equation}
 e^{z\bm e_{\bm v}\cdot\bm e_{\bm u}}
 =4\pi\sum_{l=0}^{\infty}\sum_{m=-l}^{l}
 i_l(z)(N_l^m)^2
 \widetilde Y_l^m(\Omega)\,
 \overline{\widetilde Y_l^m(\Omega_u)}.
 \label{eq:exp-Bessel-Legendre}
\end{equation}
\noindent Using \(4\pi(N_l^m)^2=(2l+1)(l-m)!/(l+m)!\) gives
\eqref{eq:fDM-Ylm-general} and its radial coefficient
\begin{equation}
 M_l^m(v)
 =B_l^m(\Omega_u)\,
   \mathcal K_l^{\mathrm{phys}}(v;u,\varsigma).
 \label{eq:Mlm-King}
\end{equation}
\noindent The zero-shift limit follows from \(i_l(0)=\delta_{l,0}\).
If the polar axis is aligned with a nonzero shift, then
\(\widetilde Y_l^m(0,\phi_u)=\delta_{m,0}\), and only the
axisymmetric terms remain.

\end{proof}

For a GMM with isotropic components
\(p=\sum_{s=1}^{S}\pi_sM_{\bm u_s,\varsigma_s}\), let \(\Omega_s\)
denote the direction of each nonzero shift.
For a centered component, the direction may be chosen
arbitrarily.
Its radial amplitudes are
\begin{equation}
 p_l^m(v)
 =\sum_{s=1}^{S}\pi_s B_l^m(\Omega_s)\,
   \mathcal K_l^{\mathrm{phys}}(v;u_s,\varsigma_s),
 \qquad u_s=|\bm u_s|.
 \label{eq:intro-GMM-modal}
\end{equation}

\noindent In \eqref{eq:intro-GMM-modal}, the shifts \(u_s\), widths
\(\varsigma_s\), weights \(\pi_s\), and directions \(\Omega_s\)
are shared across modes; the projected coefficients
\(\pi_s B_l^m(\Omega_s)\) depend on \(l,m\).

\begin{remark}[Dimensionless radial coordinate]
\label{rem:KlN-u}

For a component at the reference width, put
\(\vh=v/\varsigma\) and \(k=2u/\varsigma\).
We distinguish the dimensionless radial argument by writing
\[
 \mathcal K_l(\vh;k,\varsigma)
 :=\mathcal K_l^{\mathrm{phys}}
       (\varsigma\vh;\varsigma k/2,\varsigma).
\]
\noindent This change of radial coordinate retains the physical density
amplitude. In this coordinate, the kernel is
\begin{equation}
 \mathcal K_l(\vh;k,\varsigma)
 =\frac{2l+1}{\pi^{3/2}\varsigma^3}
   e^{-k^2/4}e^{-\vh^2}i_l(k\vh),
 \qquad k>0.
 \label{eq:KlN-u}
\end{equation}
It is the real branch of
Definition~\ref{def:King-function-Complex}.

\end{remark}

\subsection{Unified parameterization of the King kernels}
\label{subsec:main-results}

Complexifying the shift parameter places the real approximation
kernels and imaginary spectral modes in one analytic family.

\begin{definition}[Complex-parameter King function]
\label{def:King-function-Complex}
For \(l \in \Nzero\), \(\varsigma  \in \R_{> 0} \) and \(k \in \C\),
the \emph{complex-parameter King function} is
\begin{equation}
    \mathcal{K}_{l}(\vh; k, \varsigma)
    :=
    A_l (k,\varsigma)\,
    e^{-\vh^{2}}\,
    i_{l}\!\bigl(k\vh\bigr),
    \label{eq:KingN-k-C}
\end{equation}
\noindent where
\begin{equation}
    A_l (k,\varsigma)
    :=
    \frac{1}{\pi^{3/2} \varsigma^{3}}\,
    (2l+1)
    e^{-(k/2)^{2}}.
    \label{AlN}
\end{equation}
\noindent
For \(l=0\) on the real branch, \(A_0\) gives unit physical
mass, \(\mu_0[\mathcal K_0]=1\), in the convention
\eqref{eq:radial-moment}.
Proposition~\ref{prop:moment} gives the variable-width moment formula.

The two branches used below are

\begin{itemize}
    \item \textbf{Real branch} (\(k>0\)): removing the
    \(\vh\)-independent factor \(A_l\) gives the radial profile
    \begin{equation}
        \Psi_{l,k}(\vh) := e^{-\vh^{2}} i_{l}(k\vh).
        \label{eq:KingN-k}
    \end{equation}

    \item \textbf{Imaginary branch} (\(k=\ii\kappa\),
    \(\kappa>0\)): up to a nonzero scalar factor, the kernel is
    \begin{equation}
        \Phi_{l,\kappa}(\vh) := e^{-\vh^{2}} j_{l}(\kappa\vh).
        \label{eq:KingN-kappa}
    \end{equation}
\end{itemize}
\end{definition}

\begin{remark}[Zero-parameter limit and threshold renormalization]
\label{rem:King-k0}
From the small-argument asymptotics of the modified spherical Bessel
function,
\begin{equation}
    i_{l}(z) \sim \frac{z^{l}}{(2l+1)!!},
    \qquad z \to 0,
    \label{eq:King-k0-asymptotics}
\end{equation}
\noindent we obtain
\[
    \mathcal{K}_{l}(\vh;k,\varsigma)
    \sim
    A_l (k,\varsigma)\,
    \frac{k^{l}}{(2l+1)!!}\,
    \vh^{\,l} e^{-\vh^{2}},
    \qquad k \to 0.
\]
\noindent At zero shift, only the Gaussian monopole survives:
\begin{equation}
    \mathcal{K}_{l}(\vh;0,\varsigma)
    :=
    \delta_{l,0}\,
    \frac{1}{\pi^{3/2}\varsigma^{3}}\,
    e^{-\vh^{2}}.
    \label{eq:KingN-zero-mode}
\end{equation}
\noindent For the real branch, the threshold behavior is captured by the
renormalized limit
\begin{equation}
    \mathcal{K}_{l}^{\mathrm{ren}}(\vh;0,\varsigma)
    :=
    \lim_{k \to 0^{+}}
    k^{-l}\,
    \mathcal{K}_{l}(\vh;k,\varsigma)
    =
    \frac{1}{\pi^{3/2}\varsigma^{3}}\,
    \frac{\vh^{\,l}}{(2l-1)!!}\,
    e^{-\vh^{2}} .
    \label{eq:KingN-renormalized}
\end{equation}

We use \((-1)!!=1\) when \(l=0\).  Along \(k=\ii\kappa\),
the limit of \(\kappa^{-l}\mathcal K_l\) has the additional factor
\(\ii^l\).  These nonzero renormalized profiles lie outside \(\HK\):
their squared weighted norm contains
\(\int_0^\infty\vh^{2l+2}\dd\vh\).  The unrenormalized limit for
\(l>0\) is the zero function.

\end{remark}

\subsection{Reference scale and component widths}
\label{subsec:component-widths}

Fix \(\varsigma>0\) and \(\vh=v/\varsigma\).
A component has physical width \(\varsigma_s=\rho_s\varsigma\),
with \(\rho_s>0\); varying \(\rho_s\) leaves the reference coordinate
and the weights of \(H_{\mathrm{rad}}\) and
\(\mathscr H_{\mathrm{3V}}\) unchanged.
We retain the physical amplitude \(f(\varsigma\vh\,\bm e_\Omega)\),
where \(\bm e_\Omega\) is a unit direction.
A probability density with respect to \(\dd^3\hat{\bm v}\)
would include the additional Jacobian factor \(\varsigma^3\).

For \(k\ge0\) and \(\rho>0\), define the scaled real King kernel by
\begin{align}
 \mathcal K_l^{[\rho]}(\vh;k,\varsigma)
 &:=\mathcal K_l(\vh/\rho;k,\rho\varsigma)
 \notag\\
 &=\frac{(2l+1)e^{-k^2/4}}
          {\pi^{3/2}(\rho\varsigma)^3}
       e^{-\vh^2/\rho^2}i_l(k\vh/\rho).
 \label{eq:King-variable-width}
\end{align}
\noindent For a Gaussian component of width \(\rho\varsigma\) centered at
\(\bm u\), the shift parameter is
\(k=2|\bm u|/(\rho\varsigma)\).
The factor \(2\) follows from the convention
\eqref{eq:fDM}, whose covariance is half the squared width.
At \(\rho=1\), \eqref{eq:King-variable-width} is precisely the
original kernel \eqref{eq:KingN-k-C}.
At \(k=0\), it equals
\[
 \mathcal K_l^{[\rho]}(\vh;0,\varsigma)
 =\frac{\delta_{l,0}}{\pi^{3/2}(\rho\varsigma)^3}
       e^{-\vh^2/\rho^2}.
\]
\noindent Centered components therefore contribute only to \(l=0\),
independently of any assigned direction.

\subsection{Discrete KMM and Laplace--King representation}
\label{subsec:Laplace-King}

\begin{definition}[Discrete KMM and Laplace--King representation]
\label{def:discrete-Laplace-King}
Fix the reference scale \(\varsigma>0\).  For a finite component
count \(S_l^m\in\N\), the radial amplitude
\begin{equation}
 q_l^m(\vh)=\sum_{s=1}^{S_l^m}
 c_{l,s}^{m}\,
 \mathcal K_l^{[\rho_{l,s}^{m}]}(\vh;k_{l,s}^{m},\varsigma),
 \label{eq:discrete-KMM-definition}
\end{equation}
\noindent with \(c_{l,s}^{m}\in\C\), \(k_{l,s}^{m}\ge0\), and
\(\rho_{l,s}^{m}>0\), is a discrete KMM for the \((l,m)\)-mode.
For \(L\in\Nzero\), combining these amplitudes over
\(0\le l\le L\) defines a finite Laplace--King representation.
With the component direction factors written explicitly, it is
\begin{align}
    q(\vh,\theta,\phi)
    &:=
    \sum_{l=0}^{L}\sum_{m=-l}^{l}\sum_{s=1}^{S_l^m}
        \mathfrak n_{l,s}^{m}\,
        \mathcal{K}_{l}^{[\rho_{l,s}^{m}]}
            (\vh;k_{l,s}^{m},\varsigma)
        \frac{(l-m)!}{(l+m)!}\,
        \overline{\widetilde Y_l^m(\theta_{l,s}^{m},\phi_{l,s}^{m})}\,
        \widetilde Y_l^m(\theta,\phi).
    \label{eq:f-Laplace-King}
\end{align}
\noindent The effective coefficient \(c_{l,s}^{m}\) absorbs
\(\mathfrak n_{l,s}^{m}\) and the direction factor in
\eqref{eq:f-Laplace-King}.  Each retained mode admits a nonzero
direction factor.  Setting \(\rho_{l,s}^{m}=1\) gives the
fixed-width subfamily.
\end{definition}

A shared-parameter subclass uses one list of shifts, widths,
directions, and pre-direction coefficients across angular modes.
With positive probability weights, consider
\begin{equation}
 p(\bm v)=\sum_{s=1}^{S}\pi_sM_{\bm u_s,\rho_s\varsigma}(\bm v),
 \qquad \pi_s\ge0,\quad\sum_{s=1}^{S}\pi_s=1,\quad\rho_s>0.
 \label{eq:variable-width-isotropic-GMM}
\end{equation}
\noindent Applying Proposition~\ref{prop:King-kernel-from-Gaussian} to each
component gives its exact full Laplace--King expansion.
In \eqref{eq:f-Laplace-King}, this requires the untruncated angular
series and the cross-mode constraints
\[
 S_l^m=S,\qquad
 \mathfrak n_{l,s}^{m}=\pi_s,\qquad
 \rho_{l,s}^{m}=\rho_s,\qquad
 k_{l,s}^{m}=\frac{2|\bm u_s|}{\rho_s\varsigma}.
\]
\noindent Every mode uses the same direction \((\theta_s,\phi_s)\);
this direction is immaterial for a zero shift.
Gaussian normalization fixes the prefactors in
\eqref{eq:King-variable-width}.  Equal widths across components
define a further subclass.
The correspondence is exact for the full angular expansion of
isotropic Gaussian mixtures; angular truncation may lose both
exactness and nonnegativity.
Section~\ref{subsec:GMM-subclass} treats full-covariance GMMs
as approximation targets.

\section{Density of real-parameter King spans}
\label{sec:Density}

We first prove radial density at the reference width \(\rho=1\),
then extend it to finite angular expansions and countable
dictionaries.  Inclusion of the fixed-width class gives
variable-width density in the same reference spaces
(Corollary~\ref{cor:variable-width-density}).
The modal coefficients remain unrestricted.

\subsection{Radial density}
\label{subsec:radial-density}

We use the complex Gaussian-weighted radial Hilbert space
\begin{equation}
    H_{\mathrm{rad}} = L^{2}\big((0,\infty);\,\vh^{2}e^{-\vh^{2}}\,\dd\vh\big) ~. \label{eq:Hrad}
\end{equation}
\noindent Its inner product is
\begin{equation}
 \langle f,g\rangle_{H_{\mathrm{rad}}}
 :=\int_0^\infty f(\vh)\overline{g(\vh)}
      \vh^2e^{-\vh^2}\dd\vh.
 \label{eq:Hrad-inner-product}
\end{equation}

\begin{theorem}[Density of the real-parameter King span]
\label{thm:King-density}
For every \(l \in \Nzero\) the family
\[
    \mathscr{K}_{l}^{\R}
    = \bigl\{ \Psi_{l,k}(\vh) = e^{-\vh^{2}} i_{l}(k\vh)
            \mid k \in \R_{> 0} \bigr\}
\]
\noindent has dense complex linear span in \(H_{\mathrm{rad}}\).
\end{theorem}

Here totality means

\begin{equation}
 \overline{\Span\mathscr K_l^{\R}}^{\,H_{\mathrm{rad}}}=H_{\mathrm{rad}}
 \quad\Longleftrightarrow\quad
 (\mathscr K_l^{\R})^\perp=\{0\}.
 \label{eq:King-totality-meaning}
\end{equation}

Thus totality means approximation by finite linear combinations.
The behavior \(\Psi_{l,k}(\vh)=O(\vh^l)\) imposes no trace condition
at the origin on targets in this weighted \(L^2\) space.
Two lemmas prepare the proof.

\begin{lemma}[Entire function generated by a test function]
\label{lem:F-entire}
Let \(h\in H_{\mathrm{rad}}\).  For \(z\in\C\), set
\[
    \Psi_{l,z}(\vh)=e^{-\vh^2}i_l(z\vh),
\]
\noindent and define
\[
    F(z)
    =
    \int_0^\infty
    h(\vh)\Psi_{l,z}(\vh)
    \vh^2e^{-\vh^2}\dd\vh.
\]
\noindent Then \(F\) is entire and admits the expansion
\begin{equation}
    F(z)
    =
    \sum_{n=0}^{\infty}
    a_{n,l}z^{2n+l}
    \int_0^\infty
    h(\vh)\vh^{2n+l+2}e^{-2\vh^2}\dd\vh,
    \label{eq:F-series}
\end{equation}
\noindent where
\begin{equation}
    a_{n,l}
    =
    \frac{\sqrt{\pi}}
    {2^{2n+l+1}n!\Gamma(n+l+3/2)}
    \ne0 .
    \label{eq:anl}
\end{equation}
\end{lemma}

\begin{proof}
Use the absolutely convergent Taylor expansion
\[
    i_l(w)=w^l\sum_{n=0}^\infty a_{n,l}w^{2n} ,
\]
\noindent where \(a_{n,l}\) is given by Eq.~\eqref{eq:anl}.
Fix \(R>0\) and let \(|z|\le R\).  Legendre projection of the
plane-wave expansion \eqref{eq:exp-Bessel-Legendre} gives
\[
 i_l(t)=\frac12\int_{-1}^{1}e^{ts}P_l(s)\,\dd s.
\]
\noindent Since \(|P_l(s)|\le1\) on \([-1,1]\) and the coefficients
\(a_{n,l}\) are positive, for every \(\vh>0\) we have
\[
 \sum_{n=0}^{\infty}|a_{n,l}|R^{2n+l}\vh^{2n+l}
 =i_l(R\vh)\le e^{R\vh}.
\]
\noindent Hence
\[
    |h(\vh)\,i_{l}(z\vh)|\,\vh^{2}e^{-2\vh^{2}}
    \le |h(\vh)|e^{R\vh}\vh^{2}e^{-2\vh^{2}}.
\]
\noindent By the Cauchy--Schwarz inequality,
\begin{align*}
    \int_{0}^{\infty}
    |h(\vh)|e^{R\vh}\vh^{2}e^{-2\vh^{2}}\,\dd\vh
    \le \|h\|_{H_{\mathrm{rad}}}
       \Bigl(
           \int_{0}^{\infty}
           \vh^{2}
           e^{2R\vh-3\vh^{2}}\,\dd\vh
       \Bigr)^{\!1/2} 
    < \infty,
\end{align*}
\noindent since \(e^{-3\vh^{2}}\) dominates \(e^{2R\vh}\) as \(\vh\to\infty\).
The displayed integrable majorant permits termwise integration.
The resulting power series \eqref{eq:F-series} converges uniformly on
every disk \(|z|\le R\); hence \(F\) is entire.
\end{proof}

\begin{lemma}[Finite complex measure with vanishing moments]
\label{lem:moment-measure-zero}
Let \(\varepsilon\) be a finite complex Borel measure on \((0,\infty)\)
and let \(|\varepsilon|\) denote its total variation measure.  Assume
that there exists \(\eta> 0 \) such that
\begin{equation}
    \int_{0}^{\infty} e^{\eta x}\,\dd|\varepsilon|(x) < \infty.
    \label{eq:exp-moment-finite}
\end{equation}
\noindent If
\begin{equation}
    \int_{0}^{\infty} x^{n}\,\dd\varepsilon(x)=0
    \qquad\text{for all } n\in\Nzero ~.
    \label{eq:moments-vanish}
\end{equation}
\noindent then \(\varepsilon\equiv0\).
\end{lemma}

\begin{proof}
Condition \eqref{eq:exp-moment-finite} guarantees that the Laplace
transform of \(\varepsilon\),
\[
    \mathcal{L}_{\varepsilon}(t)
    := \int_{0}^{\infty} e^{tx}\,\dd\varepsilon(x),
\]
\noindent is well defined and holomorphic in the open half-plane
\(\{t \in \C : \Ree(t) < \eta\}\).
For every \(a<\eta\),
\(x^ne^{ax}\le C_{n,a,\eta}e^{\eta x}\); hence differentiation
under the integral is valid on compact subsets and gives
\[
    \mathcal{L}_{\varepsilon}^{(n)}(t)
    = \int_{0}^{\infty} x^{n} e^{tx}\,\dd\varepsilon(x),
    \qquad \Ree(t) < \eta .
\]
\noindent Evaluating at \(t=0\) and using \eqref{eq:moments-vanish} yields
\(\mathcal{L}_{\varepsilon}^{(n)}(0)=0\) for all \(n\ge 0\).
Hence \(\mathcal{L}_{\varepsilon}\) vanishes identically in a
neighborhood of the origin; by analytic continuation it vanishes on
the whole half-plane \(\Ree(t) < \eta\).

In particular, setting \(t = \ii\chi\) with \(\chi \in \R\) gives
\[
    0 = \mathcal{L}_{\varepsilon}(\ii\chi)
      = \int_{0}^{\infty} e^{\ii\chi x}\,\dd\varepsilon(x)
      = \widehat{\varepsilon}(\chi),
\]
\noindent where \(\widehat{\varepsilon}\) is the Fourier transform of the finite
complex measure \(\varepsilon\) (extended by zero to the negative
half-line).  The Fourier transform is injective on the space of finite
complex Borel measures on \(\R\) \cite{Folland2016}; therefore \(\varepsilon \equiv 0\).
\end{proof}

%--------------------------------------------------------------------
\begin{proof}[Proof of Theorem~\ref{thm:King-density}]
\label{subsec:proof-density}
Let \(g\in H_{\mathrm{rad}}\) be orthogonal to every \(\Psi_{l,k}\),
\(k>0\).  By the inner-product convention,
\[
    0
    =
    \langle \Psi_{l,k},g\rangle_{H_{\mathrm{rad}}}
    =
    \int_0^\infty
    e^{-\vh^2}i_l(k\vh)
    \overline{g(\vh)}
    \vh^2e^{-\vh^2}\dd\vh.
\]
\noindent Set \(h=\overline g\).  Then the function \(F\) in
Lemma~\ref{lem:F-entire} satisfies \(F(k)=0\) for every
\(k>0\).  The identity theorem therefore gives \(F\equiv0\)
on \(\C\).
All coefficients \(a_{n,l}\) in \eqref{eq:anl} are nonzero, so the
uniqueness of the power-series expansion \eqref{eq:F-series} yields
\begin{equation}
    M_{n}
    := \int_{0}^{\infty}
       h(\vh)\,\vh^{2n+l+2} e^{-2\vh^{2}}\,\dd\vh
    = 0,
    \qquad \forall\,n \in \Nzero.
    \label{eq:moments-zero}
\end{equation}

Now set \(x = \vh^{2}\) and define a finite complex Borel measure
\(\varepsilon\) on \((0,\infty)\) by
\[
    \dd\varepsilon(x)
    = \frac12\,h(\sqrt{x})\,x^{(l+1)/2} e^{-2x}\,\dd x .
\]
\noindent Finiteness follows from the Cauchy--Schwarz estimate
\begin{align*}
    |\varepsilon|((0,\infty))
    = \int_{0}^{\infty} |h(\vh)|\,\vh^{\,l+2} e^{-2\vh^{2}}\,\dd\vh
    \le \|h\|_{H_{\mathrm{rad}}}
       \Bigl(
           \int_{0}^{\infty} \vh^{2l+2} e^{-3\vh^{2}}\,\dd\vh
       \Bigr)^{\!1/2}
     < \infty .
\end{align*}
\noindent The moment condition \eqref{eq:moments-zero} is exactly
\(\int_{0}^{\infty} x^{n}\,\dd\varepsilon(x) = 0\) for all
\(n \in \Nzero\).

It remains to verify the exponential moment condition
\eqref{eq:exp-moment-finite} for \(\varepsilon\).
For \(\eta=1\), the left-hand side of
\eqref{eq:exp-moment-finite} becomes
\[
    \int_{0}^{\infty} e^{x}\,\dd|\varepsilon|(x)
    = \int_{0}^{\infty}
      |h(\vh)|\,\vh^{\,l+2} e^{-\vh^{2}}\,\dd\vh .
\]
\noindent Applying Cauchy--Schwarz as above gives
\begin{align*}
    \int_{0}^{\infty} |h(\vh)|\,\vh^{\,l+2} e^{-\vh^{2}}\,\dd\vh
    &\le \|h\|_{H_{\mathrm{rad}}}
       \Bigl(
           \int_{0}^{\infty} \vh^{2l+2} e^{-\vh^{2}}\,\dd\vh
       \Bigr)^{\!1/2}
     < \infty .
\end{align*}
\noindent Thus Lemma~\ref{lem:moment-measure-zero} applies (with \(\eta = 1\))
and yields \(\varepsilon \equiv 0\).

Because the density \(x^{(l+1)/2}e^{-2x}\) is strictly positive on
\((0,\infty)\), we obtain \(h(\sqrt{x}) = 0\) for almost every
\(x>0\).  Since \(x=\vh^2\) is a \(C^1\) bijection of
\((0,\infty)\) and the weight defining \(H_{\mathrm{rad}}\) is strictly
positive there, \(h=0\) almost everywhere in \(H_{\mathrm{rad}}\).
Consequently, \(g=0\).
Therefore the orthogonal complement of
\(\Span\mathscr{K}_{l}^{\R}\) is trivial, and the
span is dense in \(H_{\mathrm{rad}}\).
\end{proof}

\subsection{Density of finite Laplace--King representations}
\label{subsec:full-Laplace-King-density}

\begin{theorem}[Density of fixed-width Laplace--King representations]
\label{thm:direct-KMM-density}
Fix \(\varsigma>0\), and define
\begin{equation}
 \mathscr H_{\mathrm{3V}}
 :=L^2\!\left((0,\infty)\times\mathbb S^2;
       \vh^2e^{-\vh^2}\,\dd\vh\,\dd\Omega\right).
 \label{eq:H3V}
\end{equation}
\noindent The class of finite fixed-width Laplace--King expansions
\begin{equation}
 q(\vh,\Omega)
 =\sum_{l=0}^{L}\sum_{m=-l}^{l}\sum_{s=1}^{S_l^m}
   c_{l,s}^{m}\,\mathcal K_l(\vh;k_{l,s}^{m},\varsigma)\,
   \widetilde Y_l^m(\Omega),
 \qquad c_{l,s}^{m}\in\C,\quad k_{l,s}^{m}>0,
 \label{eq:general-discrete-KMM}
\end{equation}
\noindent with independently chosen finite parameter sets for different
\((l,m)\), is dense in \(\mathscr H_{\mathrm{3V}}\).  Here
\(c_{l,s}^{m}\) is the effective modal coefficient after the fixed
angular-direction factor is absorbed.
If the target is
real-valued, the approximants can be chosen real-valued.  If in
addition \(f\ge0\), a real approximating sequence \(q_N\) satisfies
\begin{equation}
 \|(q_N)_{-}\|_{\mathscr H_{\mathrm{3V}}}
 \le \|q_N-f\|_{\mathscr H_{\mathrm{3V}}}\longrightarrow0,
 \qquad (q_N)_{-}:=\max\{-q_N,0\}.
 \label{eq:KMM-negative-part}
\end{equation}
\end{theorem}

\begin{proof}[Proof of Theorem~\ref{thm:direct-KMM-density}]
The spherical harmonics form a complete orthogonal system in
\(L^2(\mathbb S^2)\).  By Fubini--Tonelli and boundedness of the
spherical-harmonic projections, every \(f\in\mathscr H_{\mathrm{3V}}\)
has coefficients \(f_l^m\in H_{\mathrm{rad}}\) and can first be
approximated in \(\mathscr H_{\mathrm{3V}}\) by an angularly truncated
expansion
\[
    f^{(L)}(\vh,\Omega)
    =\sum_{l=0}^{L}\sum_{m=-l}^{l}
       f_l^m(\vh)\widetilde Y_l^m(\Omega),
    \qquad f_l^m\in H_{\mathrm{rad}}.
\]
\noindent With the unnormalized convention \eqref{eq:Ylm}, Parseval's identity
reads
\begin{equation}
 \|f-f^{(L)}\|_{\mathscr H_{\mathrm{3V}}}^2
 =\sum_{l>L}\sum_{m=-l}^{l}
   \frac{\|f_l^m\|_{H_{\mathrm{rad}}}^2}{(N_l^m)^2}
 \longrightarrow0.
 \label{eq:H3V-angular-Parseval}
\end{equation}
\noindent For each of the finitely many pairs \((l,m)\),
Theorem~\ref{thm:King-density} supplies a finite sum of real-parameter
King functions that approximates \(f_l^m\) in \(H_{\mathrm{rad}}\).
Indeed, the factor \(A_l(k,\varsigma)\) in
Eq.~\eqref{eq:KingN-k-C} is nonzero for every \(k>0\), so replacing
\(\Psi_{l,k}\) by \(\mathcal K_l(\cdot;k,\varsigma)\) does not change
the linear span.  Since
\begin{equation}
    \|\widetilde Y_l^m\|_{L^2(\mathbb S^2)}^2
    =\frac{1}{(N_l^m)^2},
    \label{eq:Ylm-norm}
\end{equation}
\noindent the corresponding finite Laplace--King representation \(q^{(L)}\) satisfies
\begin{equation}
 \|f^{(L)}-q^{(L)}\|_{\mathscr H_{\mathrm{3V}}}^2
 =\sum_{l=0}^{L}\sum_{m=-l}^{l}
   \frac{\|f_l^m-q_l^m\|_{H_{\mathrm{rad}}}^2}{(N_l^m)^2}.
 \label{eq:H3V-radial-error}
\end{equation}
\noindent The finitely many radial errors can therefore be chosen so that this
sum is arbitrarily small.  The resulting
finite expression has exactly the Laplace--King form
\eqref{eq:general-discrete-KMM}, proving density.

If \(f\) is real, replace any approximant \(q\) by
\(\Ree q=(q+\overline q)/2\).  Under the convention
\eqref{eq:Ylm},
\[
 \overline{\widetilde Y_l^m}
 =(-1)^m\frac{(l+m)!}{(l-m)!}\widetilde Y_l^{-m},
\]
\noindent while the real-parameter King functions are real.  Thus
\(\Ree q\) remains a finite Laplace--King expansion, and
\(|\Ree q-f|\le|q-f|\).  Finally, when \(f\ge0\), the pointwise
inequality \((q)_{-}\le|q-f|\) for a real approximant proves
Eq.~\eqref{eq:KMM-negative-part}.
\end{proof}

\begin{remark}[Positivity and mass constraints]
\label{rem:KMM-positivity}

Pointwise nonnegativity couples the angular modes and must be
imposed separately on a finite Laplace--King approximant.
Equation~\eqref{eq:KMM-negative-part} controls its negative part in
the approximation norm.  The unweighted mass functional, defined
on integrable functions, is unbounded in that norm; exact unit
mass can be imposed through \eqref{eq:finite-KMM-normalization}.
The probability-mixture restrictions are discussed in
Section~\ref{subsec:GMM-subclass}.

\end{remark}

%--------------------------------------------------------------------

\begin{corollary}[Finite-span approximation of radial modes]
\label{cor:modal-King-density}
Fix \(\varsigma>0\) and let \(f_l^m\in H_{\mathrm{rad}}\).
For any \(\varepsilon>0\), there
exist an integer \(S\), positive parameters \(k_{1},\dots,k_{S}\), and
complex coefficients \(\mathfrak{n}_{1},\dots,\mathfrak{n}_{S}\) such that
\begin{equation}
    \Bigl\|
        f_{l}^{m}(\vh)
        - \sum_{s=1}^{S}
          \mathfrak{n}_{s}\,
          \mathcal{K}_{l}(\vh;k_{s},\varsigma)\,
          \frac{(l-m)!}{(l+m)!}\,
          \overline{\widetilde Y_l^m(\theta_s,\phi_s)}
    \Bigr\|_{H_{\mathrm{rad}}} < \varepsilon .
    \label{eq:KMM-approx}
\end{equation}
\noindent Choose directions \((\theta_s,\phi_s)\) such that
\(\widetilde Y_l^m(\theta_s,\phi_s)\neq0\).
The nonzero angular factors can be accommodated by rescaling the
coefficients \(\mathfrak n_s\).
\end{corollary}

These are the fixed-width approximants of
Definition~\ref{def:discrete-Laplace-King}, written with explicit
direction factors.  Their convergence is in \(H_{\mathrm{rad}}\).

\subsection{Countable dictionaries and projection limits}
\label{subsec:consequences-density}

Restricting the real shifts to a suitable countable set still
gives a total dictionary and convergent finite projections.

\begin{corollary}[Countable dictionary and convergent finite projections]
\label{cor:countable-KMM-projections}
Fix \(l\in\Nzero\) and \(\varsigma>0\).  Let \((k_s)_{s\ge1}\)
be a sequence of distinct positive real numbers with a finite
accumulation point, and set
\begin{equation}
 V_N=\Span\{
     \mathcal K_l(\cdot;k_s,\varsigma):1\le s\le N\}.
 \label{eq:countable-KMM-spaces}
\end{equation}
\noindent Then
\[
 V_N\subseteq V_{N+1},\qquad
 \overline{\bigcup_{N\ge1}V_N}^{\,H_{\mathrm{rad}}}=H_{\mathrm{rad}}.
\]
\noindent If \(P_N\) denotes the orthogonal projection onto \(V_N\), then,
for every \(f\in H_{\mathrm{rad}}\),
\begin{equation}
 \|f-P_Nf\|_{H_{\mathrm{rad}}}
 =\inf_{q\in V_N}\|f-q\|_{H_{\mathrm{rad}}}
 \longrightarrow0.
 \label{eq:countable-KMM-convergence}
\end{equation}
\noindent Each \(P_Nf\) is a finite discrete KMM with the prescribed parameters
\(k_1,\ldots,k_N\).
\end{corollary}

\begin{proof}
Suppose \(g\in H_{\mathrm{rad}}\) is orthogonal to every \(V_N\).
Since \(A_l(k_s,\varsigma)\ne0\), the entire function \(F\) from
Lemma~\ref{lem:F-entire}, with \(h=\overline g\), satisfies
\(F(k_s)=0\) for every \(s\).  The finite accumulation point lies in
its analytic domain, so the identity theorem gives \(F\equiv0\).
The moment-uniqueness argument in the proof of
Theorem~\ref{thm:King-density} then gives \(g=0\).  Thus the union
of the increasing spaces \(V_N\) is dense.

Each \(V_N\) is finite-dimensional and closed, so \(P_N\) exists.
Given \(f\in H_{\mathrm{rad}}\) and \(\varepsilon>0\), density provides
\(q\in V_{N_0}\) with \(\|f-q\|_{H_{\mathrm{rad}}}<\varepsilon\).
For every \(N\ge N_0\), the best-approximation property gives
\(\|f-P_Nf\|_{H_{\mathrm{rad}}}\le\|f-q\|_{H_{\mathrm{rad}}}<\varepsilon\).
\end{proof}

\begin{remark}[Projection limits and fixed-coefficient series]
\label{rem:discrete-KMM-limits}
For example, \(k_s=1+1/s\) satisfies the assumptions.
The coefficients in
\(P_Nf=\sum_{s=1}^N c_{N,s}\mathcal K_l(\cdot;k_s,\varsigma)\)
may change with \(N\).  Projection convergence does not assert
a fixed-coefficient series or a Schauder-basis expansion.

For an independently specified fixed-coefficient series, the condition
\[
 \sum_{s=1}^{\infty}|c_s|\,
 \|\mathcal K_l(\cdot;k_s,\varsigma)\|_{H_{\mathrm{rad}}}<\infty
\]
\noindent is sufficient for absolute convergence in \(H_{\mathrm{rad}}\).  If its
sum is \(F\), completeness and the triangle inequality give
\begin{equation}
 \left\|F-\sum_{s=1}^N c_s\mathcal K_l(\cdot;k_s,\varsigma)
 \right\|_{H_{\mathrm{rad}}}
 \le\sum_{s>N}|c_s|\,
 \|\mathcal K_l(\cdot;k_s,\varsigma)\|_{H_{\mathrm{rad}}}.
 \label{eq:discrete-KMM-series-tail}
\end{equation}
\noindent This condition guarantees convergence to the series sum; density alone
does not imply that every prescribed target admits such an absolutely
convergent expansion.  These discrete limits retain real parameters
and are distinct from the imaginary-branch spectral integral.
\end{remark}

\subsection{Extension to variable component widths}
\label{subsec:variable-width-density}

\begin{corollary}[Density of the variable-width class]
\label{cor:variable-width-density}
Fix \(\varsigma>0\) and the approximation spaces
\(H_{\mathrm{rad}}\) and \(\mathscr H_{\mathrm{3V}}\) in the coordinate
\(\vh=v/\varsigma\).
For every \(l\in\Nzero\), all members of
\[
 \mathscr D_l=
 \{\mathcal K_l^{[\rho]}(\cdot;k,\varsigma):
           k\ge0,\ \rho>0\}
\]
\noindent belong to \(H_{\mathrm{rad}}\), and
\[
 \overline{\Span_{\C}\mathscr D_l}^{\,H_{\mathrm{rad}}}
 =H_{\mathrm{rad}}.
\]
\noindent Consequently, the finite variable-width Laplace--King
representations of Definition~\ref{def:discrete-Laplace-King}
are dense in \(\mathscr H_{\mathrm{3V}}\).
\end{corollary}

\begin{proof}
For fixed \(l,k,\rho,\varsigma\), the estimate
\(|i_l(t)|\le e^t\) for \(t\ge0\), used in
Lemma~\ref{lem:F-entire}, gives
\[
 |\mathcal K_l^{[\rho]}(\vh;k,\varsigma)|
 \le C e^{-\vh^2/\rho^2+k\vh/\rho}.
\]
\noindent Thus its squared \(H_{\mathrm{rad}}\) norm is bounded by a constant
multiple of
\[
 \int_0^\infty \vh^2
   e^{-(1+2/\rho^2)\vh^2+2k\vh/\rho}\dd\vh<\infty.
\]
\noindent These constants may depend on the fixed kernel parameters.
The subfamily \(\rho=1,\ k>0\) is already total by
Theorem~\ref{thm:King-density}, since its normalization factors
are nonzero.  Inclusion proves radial density.
The variable-width finite three-dimensional class likewise contains
the dense fixed-width class of
Theorem~\ref{thm:direct-KMM-density}, proving the second assertion.
\end{proof}

The fixed-width dictionary remains total after adding
countably many variable-width kernels.
Projections onto nested finite spans converge in
\(H_{\mathrm{rad}}\) when every kernel is eventually included.

\section{King--Laguerre correspondence}
\label{sec:LaguerreStructure}

The fixed-width King kernels also admit a generating expansion
in generalized Laguerre polynomials.  For each \(l\), these
polynomials give the orthogonal basis

\begin{equation}
 R_{p,l}^{\mathrm{Lag}}(\vh)
 =\vh^lL_p^{(l+1/2)}(\vh^2),\qquad p\in\Nzero,
 \label{eq:Laguerre}
\end{equation}

\noindent of \(H_{\mathrm{rad}}\), with
\(\|R_{p,l}^{\mathrm{Lag}}\|_{H_{\mathrm{rad}}}^2
=\Gamma(p+l+3/2)/(2p!)\)
\cite{Szego1975, Thangavelu1993}.
The center and reference width are fixed in this comparison.

\subsection{The generating identity}
\label{subsec:proof-King-Laguerre}

\begin{theorem}[King--Laguerre expansion]
\label{thm:main-King-Laguerre}
For every \(l \in \Nzero\) and \(k \in \R_{> 0}\),
\begin{equation}
    \mathcal{K}_{l}(\vh; k, \varsigma)
    =
    \widehat M_{\varsigma}(\vh)\,
    \frac{\sqrt{\pi}}{2}
    (2l+1)
    \Bigl(\frac{k\vh}{2}\Bigr)^{\!l}
    \sum_{p=0}^{\infty}
    \frac{(-1)^{p} (k/2)^{2p}}{\Gamma(p+l+3/2)}\,
    L_{p}^{(l+1/2)}(\vh^{2}),
    \label{eq:King-Laguerre-expansion}
\end{equation}
\noindent where \(\widehat M_{\varsigma}(\vh)=\pi^{-3/2} \varsigma^{-3} e^{-\vh^{2}}\)
is the centered Gaussian amplitude in the coordinate \(\vh\),
and \(L_{p}^{(\alpha)}(\cdot)\) is the
generalized Laguerre polynomial.  The series converges absolutely and
locally uniformly for \(\vh\in[0,\infty)\).
\end{theorem}

\begin{proof}

The Bessel--Laguerre generating formula
\cite[Eq.~18.12.14]{NIST2010}, continued from \(z\) to \(-t\), reads

\begin{equation}
 e^{-t}(xt)^{-\alpha/2}I_\alpha(2\sqrt{xt})
 =\sum_{p=0}^{\infty}
   \frac{(-t)^p}{\Gamma(p+\alpha+1)}L_p^{(\alpha)}(x).
 \label{eq:Bessel-Laguerre-identity}
\end{equation}

\noindent The factor \((xt)^{-\alpha/2}I_\alpha(2\sqrt{xt})\) is defined by
the entire series
\(\sum_{n=0}^{\infty}(xt)^n/[n!\Gamma(n+\alpha+1)]\).
This specifies its value at \(xt=0\) and its complex continuation.

Apply Cauchy's estimate to the Laguerre generating function
\cite[Eq.~18.12.13]{NIST2010}.  On \(|z|=1/2\), it gives
\[
 |L_p^{(\alpha)}(x)|\le 2^{\alpha+1}e^X2^p,
 \qquad 0\le x\le X.
\]
\noindent Consequently, for \(|t|\le T\),
\[
 \sum_{p=0}^{\infty}
 \frac{|t|^p|L_p^{(\alpha)}(x)|}{\Gamma(p+\alpha+1)}
 \le 2^{\alpha+1}e^X
 \sum_{p=0}^{\infty}\frac{(2T)^p}{\Gamma(p+\alpha+1)}
 <\infty.
\]
\noindent The ratio test gives absolute, locally uniform convergence.

Set \(\alpha=l+1/2\), \(x=\vh^2\), and \(t=k^2/4\).
The relation \(i_l(z)=\sqrt{\pi/(2z)}I_{l+1/2}(z)\) then gives

\begin{equation}
 e^{-k^2/4}i_l(k\vh)
 =\frac{\sqrt{\pi}}{2}\Bigl(\frac{k\vh}{2}\Bigr)^l
   \sum_{p=0}^{\infty}
   \frac{(-1)^p(k/2)^{2p}}{\Gamma(p+l+3/2)}
   L_p^{(l+1/2)}(\vh^2).
 \label{eq:King-raw-expansion}
\end{equation}

Multiplication by
\((2l+1)\pi^{-3/2}\varsigma^{-3}e^{-\vh^2}\)
proves \eqref{eq:King-Laguerre-expansion}, including its continuous
extension to \(\vh=0\).

\end{proof}

\subsection{Weighted-space correspondence}
\label{subsec:renormalized-Laguerre}

Dividing by \(\vh^l\) relates the radial density theorem to the
one-variable Laguerre space.

\begin{corollary}[Density of the radially rescaled King family]
\label{cor:King-density-renormalization}
For every \(l \in \Nzero\), the renormalized family
\[
    \widetilde{\mathscr{K}}_{l}
    = \bigl\{
        \widetilde{\Psi}_{l,k}(\vh)
        = \vh^{-l} e^{-\vh^{2}} i_{l}(k\vh)
        \mid k \in \R_{> 0}
      \bigr\}
\]
\noindent has a dense linear span in the weighted space
\begin{equation}
    \widetilde{H}_{\mathrm{rad}}^{\,l}
    = L^{2}\big((0,\infty);\, \vh^{2l+2} e^{-\vh^{2}}\,\dd\vh\big).
    \label{eq:Hrad-renormalization}
\end{equation}
\end{corollary}

\begin{proof}
The multiplication operator
\(V : H_{\mathrm{rad}} \to \widetilde{H}_{\mathrm{rad}}^{\,l}\),
\((Vf)(\vh) = \vh^{-l} f(\vh)\),
has inverse \((V^{-1}g)(\vh)=\vh^l g(\vh)\) and is unitary because
for all \(f,g\in H_{\mathrm{rad}}\),
\[
    \langle Vf, Vg \rangle_{\widetilde{H}_{\mathrm{rad}}^{\,l}}
    = \int_{0}^{\infty}
      \vh^{-l} f(\vh)\,
      \overline{\vh^{-l} g(\vh)}\,
      \vh^{2l+2} e^{-\vh^{2}}\,\dd\vh
    = \langle f, g \rangle_{H_{\mathrm{rad}}} .
\]
\noindent By Theorem~\ref{thm:King-density},
\(\Span\mathscr{K}_{l}^{\R}\) is dense in
\(H_{\mathrm{rad}}\).  Its image under \(V\) is
\(\Span\widetilde{\mathscr K}_l\), which is therefore
dense in \(\widetilde H_{\mathrm{rad}}^{\,l}\).
\end{proof}

More precisely, the map
\(g(\vh)\mapsto 2^{-1/2}g(\sqrt{x})\) is unitary from
\(\widetilde H_{\mathrm{rad}}^{\,l}\) onto the Laguerre space
\[
 L^2((0,\infty);x^{l+1/2}e^{-x}\dd x).
\]
\noindent It sends the renormalized King family to scalar multiples of
\(x^{-l/2}e^{-x}i_l(k\sqrt{x})\), whose span is therefore dense in
that Laguerre space.

\section{The King operator and its spectral resolution}
\label{sec:KingOperator}

The fixed-width King equation has Sturm--Liouville form with
the weight defining \(\HK\) in \eqref{eq:HK}.
A Liouville transformation reduces it to the free radial
Schr\"odinger operator, giving its self-adjoint realization
and spectral resolution.

\begin{definition}[King differential expression]
\label{def:tau}
For \(l \in \Nzero\) and \(\vh > 0\), the \emph{King
differential expression} is
\begin{equation}
    \tau_{l} f(\vh)
    = -\frac{1}{\omega(\vh)}
      \frac{\dd}{\dd\vh}
      \Bigl(
          \omega(\vh)\,\frac{\dd f}{\dd\vh}
      \Bigr)
    + V_{l}(\vh)\,f(\vh),
    \label{eq:tau}
\end{equation}
\noindent where \(\omega(\vh) = \vh^{2}e^{2\vh^{2}}\) and
\(V_{l}(\vh) = \frac{l(l+1)}{\vh^{2}} - (4\vh^{2}+6)\).
The associated Hilbert space is
\begin{equation}
    \HK = L^{2}\big((0,\infty);\,\omega(\vh)\,\dd\vh\big) ~. \label{eq:HK}
\end{equation} 
\end{definition}

\subsection{The King equation}
\label{subsec:King-eq}

\begin{proposition}[King equation]
\label{prop:King-equation}
For every \(l \in \Nzero\) and \(k \in \C\), the
function \(e^{-\vh^{2}}i_{l}(k\vh)\) satisfies the second-order linear
ordinary differential equation
\begin{equation}
    f''(\vh)
    + \Bigl( \frac{2}{\vh} + 4\vh \Bigr) f'(\vh)
    + \Bigl[ -\frac{l(l+1)}{\vh^{2}} + 4\vh^{2} + 6 - k^{2} \Bigr] f(\vh)
    = 0 .
    \label{eq:Kingeq}
\end{equation}
\noindent We refer to \eqref{eq:Kingeq} as the \emph{King equation}.
\end{proposition}

\begin{proof}
Set \(z = k\vh \in \C \).  The modified spherical Bessel function satisfies \cite[\S10.47]{NIST2010}
\begin{equation}
    z^{2} i_{l}''(z) + 2z i_{l}'(z)
    - \bigl[ z^{2} + l(l+1) \bigr] i_{l}(z) = 0 .
    \label{eq:mod-sphere-Bessel}
\end{equation}
\noindent With \(f(\vh) = e^{-\vh^{2}} i_{l}(k\vh)\) we compute
\begin{align*}
    f'(\vh)  &= e^{-\vh^{2}} \bigl[ k i_{l}'(z) - 2\vh i_{l}(z) \bigr], \\
    f''(\vh) &= e^{-\vh^{2}} \bigl[ k^{2} i_{l}''(z) - 4k\vh i_{l}'(z)
                + (4\vh^{2} - 2) i_{l}(z) \bigr].
\end{align*}
\noindent Substituting into the left-hand side of \eqref{eq:Kingeq} and
simplifying,
\begin{align*}
    &f''(\vh) + \Bigl(\frac{2}{\vh} + 4\vh\Bigr) f'(\vh)
      + \Bigl[ -\frac{l(l+1)}{\vh^{2}} + 4\vh^{2} + 6 - k^{2} \Bigr] f(\vh) \\
    &= e^{-\vh^{2}} \Bigl\{
          k^{2} i_{l}''(z)
          + \frac{2k}{\vh} i_{l}'(z)
          - \Bigl[ k^{2} + \frac{l(l+1)}{\vh^{2}} \Bigr] i_{l}(z)
        \Bigr\} \\
    &= \frac{e^{-\vh^{2}}}{\vh^{2}} \Bigl\{
          z^{2} i_{l}''(z) + 2z i_{l}'(z)
          - \bigl[ z^{2} + l(l+1) \bigr] i_{l}(z)
        \Bigr\}.
\end{align*}
\noindent The bracket vanishes by \eqref{eq:mod-sphere-Bessel},
which proves \eqref{eq:Kingeq}.
\end{proof}

%--------------------------------------------------------------------
\subsection{Unitary equivalence and self-adjointness}
\label{subsec:unitary}

Multiplying \eqref{eq:Kingeq} by
\(\omega(\vh)=\vh^2e^{2\vh^2}\) gives
\(\tau_lf=-k^2f\).  A self-adjoint realization requires an endpoint
condition when \(l=0\).  Define the maximal domain by

\begin{equation}
 \Dom_{\max}(\tau_l)
 :=\{f\in\HK:
       f,\omega f'\in\ACloc(0,\infty),\ \tau_lf\in\HK\}.
 \label{eq:Dmax}
\end{equation}

\begin{theorem}[Unitary equivalence and self-adjointness]
\label{thm:main-unitary}

The map \((\mathcal Uf)(\vh)=\vh e^{\vh^2}f(\vh)\) is unitary
from \(\HK\) onto \(L^2(0,\infty)\).  On \(C_c^\infty(0,\infty)\),

\[
 \mathcal U\tau_l\mathcal U^{-1}
 =-\frac{\dd^2}{\dd\vh^2}+\frac{l(l+1)}{\vh^2}.
\]

Let \(h_l\) be the unique self-adjoint realization of this
differential expression for \(l\ge1\), and its Dirichlet realization
\(\psi(0)=0\) for \(l=0\).  Then
\(\mathcal L_l=\mathcal U^{-1}h_l\mathcal U\) is self-adjoint on
\(\HK\), with

\[
 \Dom(\mathcal L_l)=
 \begin{cases}
  \Dom_{\max}(\tau_l),&l\ge1,\\
  \{f\in\Dom_{\max}(\tau_0):
       \lim_{\vh\downarrow0}\vh f(\vh)=0\},&l=0.
 \end{cases}
\]
\end{theorem}

\begin{proof}

The inverse is
\((\mathcal U^{-1}\psi)(\vh)=\vh^{-1}e^{-\vh^2}\psi(\vh)\),
and

\[
 \|\mathcal Uf\|_{L^2}^2
 =\int_0^\infty |f(\vh)|^2\vh^2e^{2\vh^2}\dd\vh
 =\|f\|_{\HK}^2.
\]

Direct differentiation gives

\[
 \tau_l(\vh^{-1}e^{-\vh^2}\psi)
 =\vh^{-1}e^{-\vh^2}
   \left(-\psi''+\frac{l(l+1)}{\vh^2}\psi\right).
\]

The smooth nonvanishing multiplier also maps the maximal domains
onto each other.  Near the origin, the zero-energy solutions of the
transformed equation are \(\vh^{l+1}\) and \(\vh^{-l}\).
Both are square-integrable there only for \(l=0\); hence the origin
is limit-circle for \(l=0\) and limit-point for \(l\ge1\).
Infinity is limit-point for every \(l\)
\cite{Zettl2005, Teschl2014}.
Weyl's classification therefore gives uniqueness for \(l\ge1\).
For \(l=0\), the Dirichlet condition on
\(\psi=\mathcal Uf\) is equivalent to
\(\lim_{\vh\downarrow0}\vh f(\vh)=0\), since \(e^{\vh^2}\to1\).
Unitary conjugation gives the stated self-adjoint domains.

\end{proof}

\subsection{Spectral resolution}
\label{subsec:spectrum}

For \(\psi\in C_c^\infty(0,\infty)\), define

\begin{equation}
 (\mathscr H_l\psi)(\kappa)
 :=\sqrt{\frac{2}{\pi}}\,\kappa
   \int_0^\infty\psi(\vh)\vh j_l(\kappa\vh)\dd\vh.
 \label{eq:normalized-Hankel}
\end{equation}
\noindent Its kernel is
\(\sqrt{\kappa\vh}\,J_{l+1/2}(\kappa\vh)\).
Hankel inversion extends \(\mathscr H_l\) to a self-inverse unitary
operator on \(L^2(0,\infty)\)
\cite[Eqs.~10.22.76--10.22.77]{NIST2010}.

\begin{theorem}[Spectral resolution]
\label{thm:main-spectral}

For every \(l\in\Nzero\), the operator \(\mathcal L_l\) has purely
absolutely continuous spectrum

\[
 \sigma(\mathcal L_l)=\sigma_{\mathrm{ac}}(\mathcal L_l)=[0,\infty),
 \qquad
 \sigma_{\mathrm p}(\mathcal L_l)=
 \sigma_{\mathrm{sc}}(\mathcal L_l)=\varnothing.
\]

For the modes \(\Phi_{l,\kappa}\) defined in
\eqref{eq:KingN-kappa}, set
\(\mathcal N_l(\kappa)=\pi/(2\kappa^2)\).  The coefficient map

\begin{equation}
 \widetilde f(\kappa)
 :=\sqrt{\frac{2}{\pi}}\,\kappa
       [\mathscr H_l(\mathcal Uf)](\kappa)
 \label{eq:King-coefficient-L2}
\end{equation}

is unitary from \(\HK\) onto
\(L^2((0,\infty);\mathcal N_l(\kappa)\dd\kappa)\).
Its inverse and Parseval identity are

\begin{align}
 f(\vh)&=\int_0^\infty
       \widetilde f(\kappa)\Phi_{l,\kappa}(\vh)\dd\kappa,
 \label{eq:KingTran}\\
 \|f\|_{\HK}^2&=\int_0^\infty
       |\widetilde f(\kappa)|^2\mathcal N_l(\kappa)\dd\kappa.
 \label{eq:parseval}
\end{align}

Inversion is understood in the \(\HK\)-norm, as the spectral
cutoffs tend to \(0\) and \(\infty\).
For \(f\in C_c^\infty(0,\infty)\), the analysis formula is

\begin{equation}
 \widetilde f(\kappa)
 =\frac{2\kappa^2}{\pi}
   \int_0^\infty f(\vh)\Phi_{l,\kappa}(\vh)\omega(\vh)\dd\vh.
 \label{eq:KingTranInv}
\end{equation}

The kernels satisfy
\(\tau_l\Phi_{l,\kappa}=\kappa^2\Phi_{l,\kappa}\) and the
distributional orthogonality relation

\begin{equation}
 \int_0^\infty
 \Phi_{l,\kappa}(\vh)\Phi_{l,\kappa'}(\vh)\omega(\vh)\dd\vh
 =\mathcal N_l(\kappa)\delta(\kappa-\kappa').
 \label{eq:King-distributional-orthogonality}
\end{equation}

They are generalized eigenfunctions, not elements of \(\HK\).
There is no spectral atom at \(\kappa=0\).

\end{theorem}

\begin{proof}

Let \(M_{\kappa^2}\) be multiplication by \(\kappa^2\) on
\(L^2(0,\infty)\), with
\(\Dom(M_{\kappa^2})=\{g:\ g,\kappa^2g\in L^2(0,\infty)\}\).
The spherical-Bessel equation and integration by parts on
\(C_c^\infty(0,\infty)\) give
\(\mathscr H_l(h_l\psi)=\kappa^2\mathscr H_l\psi\).
Thus \(\mathscr H_l^{-1}M_{\kappa^2}\mathscr H_l\) is a
self-adjoint extension of the minimal radial operator.
For \(l\ge1\), essential self-adjointness identifies it with \(h_l\).

For \(l=0\), \(\mathscr H_0\) is the Fourier sine transform.
The odd extension of a function on the half-line belongs to
\(H^2(\R)\) precisely when the function belongs to \(H^2(0,\infty)\)
and has zero trace.  Fourier transformation therefore gives

\[
 \mathscr H_0^{-1}\Dom(M_{\kappa^2})
 =H^2(0,\infty)\cap H_0^1(0,\infty).
\]

The pulled-back operator acts as \(-\psi''\) on this domain and
is exactly the selected Dirichlet realization.  Consequently,
for every \(l\ge0\),

\[
 \mathscr H_lh_l\mathscr H_l^{-1}=M_{\kappa^2}.
\]

The multiplication operator has spectrum \([0,\infty)\).
For \(g\in L^2(0,\infty)\), changing variables from \(\kappa\) to
\(\lambda=\kappa^2\) gives its scalar spectral measure the density
\(|g(\sqrt\lambda)|^2/(2\sqrt\lambda)\) for \(\lambda>0\), with
no atom at zero.  This proves the spectral assertion.

The coefficient map in \eqref{eq:King-coefficient-L2} is the
composition of \(\mathcal U\), \(\mathscr H_l\), and the unitary
multiplication map from \(L^2(\dd\kappa)\) to
\(L^2(\mathcal N_l\dd\kappa)\).  Hankel inversion and Plancherel give
\eqref{eq:KingTran} and \eqref{eq:parseval}.
Substituting \(\mathcal Uf=\vh e^{\vh^2}f\) into
\eqref{eq:normalized-Hankel} gives \eqref{eq:KingTranInv} on the
stated dense test space; \eqref{eq:King-coefficient-L2} defines its
extension to all of \(\HK\).

Finally, the King equation at \(k=\ii\kappa\) gives the
generalized eigenfunction equation.  Cancellation of the Gaussian
factors in \eqref{eq:King-distributional-orthogonality} reduces it to
spherical-Bessel orthogonality,

\[
 \int_0^\infty\vh^2j_l(\kappa\vh)j_l(\kappa'\vh)\dd\vh
 =\frac{\pi}{2\kappa^2}\delta(\kappa-\kappa'),
\]

interpreted as a distribution on \((0,\infty)^2\)
\cite{Watson1944, NIST2010}.

\end{proof}

\subsection{Continuous KMM and spectral representation}
\label{subsec:continuous-KMM}

For \(k=\ii\kappa\), the relation
\(i_l(\ii\kappa\vh)=\ii^lj_l(\kappa\vh)\) gives

\begin{equation}
 \mathcal K_l(\vh;\ii\kappa,\varsigma)
 =a_l(\kappa)\Phi_{l,\kappa}(\vh),\qquad
 a_l(\kappa)=
 \frac{(2l+1)e^{\kappa^2/4}}{\pi^{3/2}\varsigma^3}\,\ii^l.
 \label{eq:King-branch-factor}
\end{equation}

For \(f_l^m\in\HK\), set
\(b_l^m=\widetilde{f_l^m}/a_l\), where
\(\widetilde{f_l^m}\) is defined by
\eqref{eq:King-coefficient-L2}.  The resulting spectral integral

\begin{equation}
 f_l^m(\vh)=\int_0^\infty b_l^m(\kappa)
       \mathcal K_l(\vh;\ii\kappa,\varsigma)\dd\kappa
 \label{eq:app-forward}
\end{equation}

\noindent defines the fixed-width \emph{continuous KMM}.
The truncated spectral integrals converge in \(\HK\), and the coefficient
transform diagonalizes \(\mathcal L_l\).
For \(f_l^m\in C_c^\infty(0,\infty)\), the coefficients are

\begin{equation}
 b_l^m(\kappa)=\frac{2\kappa^2}{\pi|a_l(\kappa)|^2}
   \int_0^\infty f_l^m(\vh)
   \overline{\mathcal K_l(\vh;\ii\kappa,\varsigma)}
   \omega(\vh)\dd\vh.
 \label{eq:app-inverse}
\end{equation}

The conjugate kernel contributes \(\overline{a_l}\), giving
\(b_l^m=\widetilde{f_l^m}/a_l\) as required.
For every \(f_l^m\in\HK\), Parseval takes the form

\begin{equation}
 \|f_l^m\|_{\HK}^2
 =\int_0^\infty|b_l^m(\kappa)|^2|a_l(\kappa)|^2
      \frac{\pi}{2\kappa^2}\dd\kappa.
 \label{eq:continuous-KMM-Parseval}
\end{equation}

If
\(\sum_{l=0}^{\infty}\sum_{m=-l}^{l}
 \|f_l^m\|_{\HK}^2/(N_l^m)^2<\infty\),
angular Parseval gives the continuous Laplace--King representation

\begin{equation}
 f(\vh,\Omega)=
 \sum_{l=0}^{\infty}\sum_{m=-l}^{l}
 \left[\int_0^\infty b_l^m(\kappa)
   \mathcal K_l(\vh;\ii\kappa,\varsigma)\dd\kappa\right]
 \widetilde Y_l^m(\Omega).
 \label{eq:continuous-Laplace-King}
\end{equation}

The angular sums and spectral inversions converge in
\(L^2((0,\infty)\times\mathbb S^2;
       \vh^2e^{2\vh^2}\dd\vh\dd\Omega)\).
The spaces satisfy
\(\|g\|_{H_{\mathrm{rad}}}\le\|g\|_{\HK}\), and the inclusion is
strict: \(e^{-\vh^2}\in H_{\mathrm{rad}}\) but
\(e^{-\vh^2}\notin\HK\).
The continuous KMM uses the fixed-width operator; variable-width
kernels enter the approximation results of
Section~\ref{sec:Density}.

\section{Approximation constraints and applications}
\label{sec:Applications}

We relate finite Laplace--King approximation to probability
constraints, angular and radial errors, and moment evaluation,
then outline a Fokker--Planck application.

\subsection{Gaussian-mixture targets and probability constraints}
\label{subsec:GMM-subclass}

Let \(p\) be a finite GMM of the form \eqref{eq:intro-GMM} in
three dimensions, with positive-definite component covariances.
The Gaussian formula \eqref{eq:intro-Gaussian} gives
\begin{equation}
 0\le p(\bm v)\le C_p:=
 \sum_{s=1}^{S}
 \frac{\pi_s}{(2\pi)^{3/2}(\det\bm\Sigma_s)^{1/2}}<\infty.
 \label{eq:GMM-bound}
\end{equation}
\noindent In the fixed reference coordinate, write
\(p_\varsigma(\vh,\Omega)=p(\varsigma\vh\bm e_\Omega)\).
The weight in \eqref{eq:H3V} is integrable, so
\begin{equation}
 \|p_\varsigma\|_{\mathscr H_{\mathrm{3V}}}^{2}
 \le 4\pi C_p^2\int_0^\infty\vh^2e^{-\vh^2}\dd\vh
 =\pi^{3/2}C_p^2.
 \label{eq:GMM-H3V}
\end{equation}
\noindent Theorem~\ref{thm:direct-KMM-density} therefore gives real
finite Laplace--King approximants to every finite full-covariance
GMM in this norm.  The exact componentwise King expansion in
\eqref{eq:intro-GMM-modal} applies to isotropic components.

A real-valued finite approximation \(q\) defines a probability
density in physical velocity when it satisfies
\begin{equation}
 q(\vh,\Omega)\ge0,\qquad \mu_0[q_0^0]=1.
 \label{eq:LK-probability-constraints}
\end{equation}
\noindent Absolute integrability follows from
Proposition~\ref{prop:moment} at \(j=0\).
Equations~\eqref{eq:monopole-total-integral}
and~\eqref{eq:finite-KMM-normalization} then identify unit mass
with a coefficient equality.
Nonnegativity is a condition on the full angular sum and
couples its modes; the norm estimate
\eqref{eq:KMM-negative-part} does not enforce it.

Angular parity does not determine the sign of a modal amplitude.
For example,
\(f(\vh,\Omega)=a(\vh)[1-\eta P_2(\cos\theta)]\), with
\(a\ge0\) and \(0<\eta<1\), is nonnegative, whereas its even-degree
amplitude is \(f_2^0=-\eta a\).
Even in the axisymmetric case, nonnegative radial amplitudes
do not ensure a nonnegative reconstruction, since
\(P_l(\cos\theta)\) changes sign for \(l>0\).
The constraints \eqref{eq:LK-probability-constraints} must
therefore be imposed on the reconstructed function.
The positive isotropic GMM in
\eqref{eq:variable-width-isotropic-GMM} gives a sufficient
shared-parameter construction before angular truncation.

\begin{remark}[Fixed-width positive mixtures]
For a fixed component width, every finite or countable positive mixture

\[
 g(\bm v)=\sum_{s=1}^{\infty}\pi_sM_{\bm u_s,\varsigma}(\bm v),
 \qquad \pi_s\ge0,\quad\sum_{s=1}^{\infty}\pi_s=1,
\]
\noindent satisfies \(g(\bm v)\le\pi^{-3/2}\varsigma^{-3}\).
For \(0<\tau<\varsigma\), the Gaussian \(M_{\bm0,\tau}\) exceeds
this bound on a ball of positive weighted measure.  Its
\(\mathscr H_{\mathrm{3V}}\)-distance from this subclass is therefore
strictly positive, with \(\vh=|\bm v|/\varsigma\) fixed.
Allowing component widths to vary removes this obstruction and
permits the classical \(L^1\) density result
\cite{Nguyen2023LocationScale}.

\end{remark}

\begin{remark}[Scope of projection convergence]
\label{rem:mixture-convergence-types}
Corollary~\ref{cor:countable-KMM-projections} concerns deterministic
best approximation.  Density-estimation consistency requires a
specified statistical procedure
\cite{GhosalGhoshRamamoorthi1999,ShenTokdarGhosal2013};
parameter recovery also requires identifiability
\cite{Teicher1963Identifiability}.
Convergence of a fitting iteration depends on its objective and
algorithm \cite{DempsterLairdRubin1977,Wu1983EM}.

\end{remark}

\subsection{Angular and radial approximation errors}
\label{subsec:approximation-budget}

For \(f\in\mathscr H_{\mathrm{3V}}\), let
\(q^{(L)}=\sum_{l=0}^{L}\sum_{m=-l}^{l}q_l^m\widetilde Y_l^m\),
where each \(q_l^m\in H_{\mathrm{rad}}\) is a discrete KMM with
fixed or variable component widths.
Orthogonality between retained and discarded angular modes gives
the exact identity
\begin{align}
 \|f-q^{(L)}\|_{\mathscr H_{\mathrm{3V}}}^2
 =\sum_{l>L}\sum_{m=-l}^{l}
       \frac{\|f_l^m\|_{H_{\mathrm{rad}}}^2}{(N_l^m)^2}
       + 
       \sum_{l=0}^{L}\sum_{m=-l}^{l}
       \frac{\|f_l^m-q_l^m\|_{H_{\mathrm{rad}}}^2}{(N_l^m)^2}.
 \label{eq:KMM-total-error-budget}
\end{align}

\noindent Apply Parseval to
\(f-q^{(L)}=(f-f^{(L)})+(f^{(L)}-q^{(L)})\).
The two sums measure the discarded angular tail and the radial
error in the retained modes.  The factors \((N_l^m)^{-2}\)
account for the harmonic normalization.
There is no corresponding sum of squares of the nonorthogonal
King coefficients.

The cutoff \(L\) retains \((L+1)^2\) modes and
\(\sum_{l=0}^{L}\sum_{m=-l}^{l}S_l^m\) radial terms.
Equation~\eqref{eq:KMM-total-error-budget} separates their error
contributions without prescribing component counts.
Computational cost depends on angular decay and radial resolution.

\subsection{Moment evaluation and constraints}
\label{subsec:moment-evaluation}

Angularly resolved moments reduce to radial integrals.
Their zeroth-order monopole is the total integral, which gives
the normalization condition.
Suppose \(f(v,\cdot)\in L^2(\mathbb S^2)\) for almost every \(v>0\)
and the integral below is absolutely convergent.
With \(\vh=v/\varsigma\), let \(f_l^m(\vh)\) denote the angular
coefficients of \(f(\varsigma\vh\bm e_\Omega)\) under
\eqref{eq:f-radial-coefficient}.  For \(j\in\R\), define

\begin{align}
 \mathcal Q_{j,l}^{m}[f]
 &:=\int_{\R^3}|\bm v|^j
       \overline{\widetilde Y_l^m(\Omega)}f(\bm v)\,\dd^3\bm v
 \notag\\
 &=\frac{\varsigma^{j+3}}{(N_l^m)^2}
       \int_0^\infty\vh^{j+2}f_l^m(\vh)\,\dd\vh
  =\frac{\mu_j[f_l^m]}{4\pi(N_l^m)^2},
 \label{eq:angularly-resolved-moment}
\end{align}
\noindent where \(\mu_j\) is defined in \eqref{eq:radial-moment}.
Proposition~\ref{prop:moment} proves that a King component with
\(k>0\) has a finite absolute radial moment of order \(j\)
if and only if \(j>-l-3\).
This inequality ensures termwise integrability of a finite KMM,
whose moments are then evaluated by \eqref{eq:moment-Klm} and
\eqref{eq:moment-mixture}.  Cancellations between individually
divergent terms require separate treatment.
With shifts and widths fixed, prescribed moments give linear
constraints on the modal coefficients.

For an integrable \(f\), the special case \(j=l=m=0\) gives
the total integral.  Since \(\widetilde Y_0^0=1\) and
\(4\pi(N_0^0)^2=1\), \eqref{eq:angularly-resolved-moment} reduces to

\begin{equation}
 \mathcal Q_{0,0}^{0}[f]
 =\mu_0[f_0^0]
 =\int_{\R^3}f(\bm v)\,\dd^3\bm v.
 \label{eq:monopole-total-integral}
\end{equation}
\noindent Proposition~\ref{prop:moment} and the centered case give
\(\mu_0[\mathcal K_0^{[\rho]}]=1\) for all \(k\ge0\) and
\(\rho>0\).  For the monopole \(q_0^0\) of a finite Laplace--King
representation, \eqref{eq:moment-mixture-M000} therefore makes
unit-mass normalization equivalent to

\begin{equation}
 \mu_0[q_0^0]
 =\sum_{s=1}^{S_0^0}c_{0,s}^{0}
 =1.
 \label{eq:finite-KMM-normalization}
\end{equation}
\noindent Normalization is independent of component shifts and widths.
Together with reality and nonnegativity, it completes the
probability constraints \eqref{eq:LK-probability-constraints}.

\begin{remark}[Convergence of moments]
\label{rem:moment-convergence}
Let \(j\in\R\) and \(g_n\to g\) in \(H_{\mathrm{rad}}\).
A sufficient condition for \(\mu_j[g_n]\to\mu_j[g]\) is
\begin{align}
 \lim_{\delta\downarrow0}\sup_n
 \int_0^\delta\vh^{j+2}|g_n(\vh)|\dd\vh&=0,
 \notag\\
 \lim_{R\to\infty}\sup_n
 \int_R^\infty\vh^{j+2}|g_n(\vh)|\dd\vh&=0.
 \label{eq:moment-endpoint-control}
\end{align}
\noindent On each compact annulus \([\delta,R]\),
Cauchy--Schwarz gives
\[
 \int_\delta^R\vh^{j+2}|g_n-g|\dd\vh
 \le
 \left(\int_\delta^R\vh^{2j+2}e^{\vh^2}\dd\vh\right)^{1/2}
 \|g_n-g\|_{H_{\mathrm{rad}}}.
\]
\noindent An almost-everywhere convergent subsequence and Fatou's lemma
\cite{Axler2020Measure}
transfer the endpoint bounds to \(g\).
Combining the compact-annulus estimate with the two endpoint
bounds gives convergence in
\(L^1((0,\infty);\vh^{j+2}\dd\vh)\), hence moment convergence.
\end{remark}

\subsection{Fokker--Planck equations}
\label{subsec:kinetic-example}

Tabandeh et al.\ \cite{Tabandeh2022PhysicsBasedMixtures} use
Gaussian mixtures to approximate evolving probability densities
in Fokker--Planck equations.
We plan to study a KMM formulation for three-dimensional
velocity distributions in plasma kinetic theory
\cite{Thomas2012,wang2024Aconservative,wang2025general}.
With prescribed coefficients, spatially homogeneous evolution
takes the form

\begin{equation}
 \frac{\partial f(\bm v,t)}{\partial t}=\mathcal C[f],
 \label{eq:FP-eq}
\end{equation}
\noindent with a Fokker--Planck operator in divergence form
\begin{equation}
 \mathcal C[f]=\nabla_{\bm v}\cdot
 \left[\bm A(\bm v)f+\bm D(\bm v)\cdot\nabla_{\bm v}f\right],
 \label{eq:FP-operator}
\end{equation}
\noindent where \(\bm A\) is the friction vector and \(\bm D\) is a symmetric
positive semidefinite diffusion tensor
\cite{Risken1989,Gardiner2009,Villani2002}.
At each fixed time, the density theorem applies when the
distribution belongs to \(\mathscr H_{\mathrm{3V}}\).
The moment identities in Section~\ref{subsec:moment-evaluation}
give normalization and other integral constraints on a finite
representation, subject to termwise integrability.

Comparison with the Gaussian-mixture method will require
analysis of \(\mathcal C\) on the KMM ansatz, including
nonnegativity, operator-specific moment balances, and numerical
stability and convergence.
The King operator \(\mathcal L_l\) enters only through the
radial representation.

\section{Conclusions}
\label{sec:Conclusion}

Laplace--King representations allow the radial parameters to be
chosen separately for each angular mode, while retaining the
King kernels of isotropic Gaussian mixtures.  We have shown that
complex linear combinations of fixed-width King kernels with
positive real shifts are dense in the radial weighted \(L^2\)
space.  Together with angular Parseval, this establishes density
of finite three-dimensional representations.  Allowing variable
widths retains this property in the same reference norm.

The same analytic kernel family also connects approximation
with spectral analysis.  The King--Laguerre identity gives a
generating expansion of individual kernels, while the imaginary
branch yields the continuous KMM through the spectral resolution
of the self-adjoint King operator.  This spectral representation
holds in a different weighted Hilbert space.

The error decomposition separates angular truncation from
radial approximation.  For fixed kernel parameters, the moment
formulas express normalization and other prescribed moments as
coefficient constraints.
Numerical development requires suitable choices of shifts and
widths and stable recovery of the coefficients.
The planned Fokker--Planck study must also address nonnegativity,
operator-specific moment balances, and convergence of the
time-dependent approximation.

\section*{Acknowledgments}

The authors are grateful to Dr. Kunyu Chen and Prof. Long Zeng for
valuable discussions and helpful suggestions.  This work was supported
by the National Natural Science Foundation of China (NSFC) under Grant
No.~12335014 and Beijing Primal Energy Theory Research Institute.

\section*{Declarations}
\noindent\textbf{Conflict of interest.}
The authors declare that they have no conflict of interest.

\appendix

\section{\texorpdfstring{Moments of variable-width King kernels and mixtures}{Moments of variable-width King kernels and mixtures}}
\label{app:moment-formula}

This appendix derives the variable-width moments used in
Section~\ref{subsec:moment-evaluation}.
The reference scale \(\varsigma>0\) and coordinate
\(\vh=v/\varsigma\) remain fixed; a component has width
\(\rho\varsigma\).

\subsection{Moment convention and existence}

For \(j\in\R\) and a radial amplitude \(g\) satisfying
\(\int_0^\infty\vh^{j+2}|g(\vh)|\dd\vh<\infty\), define
\begin{equation}
 \mu_j[g]:=4\pi\varsigma^{j+3}
       \int_0^\infty\vh^{j+2}g(\vh)\dd\vh.
 \label{eq:radial-moment}
\end{equation}
\noindent For a normalized per-particle velocity distribution in the kinetic
application, the mass-weighted modal convention is
\begin{equation}
 \mathcal M_{j,l}^{m}[f]
 :=m_a n_a\,\mu_j[f_l^m],
 \label{eq:moment-def}
\end{equation}
\noindent where \(m_a>0\) is the particle mass and \(n_a>0\) the prescribed
number density.  These factors enter only the physical convention.
Since \(\widetilde Y_0^0=1\), \(\mu_0[f_0^0]\) is the integral of
an integrable distribution over physical velocity space.
Other physical moments involve the corresponding angular
contractions and, for energy, a factor \(1/2\).

\begin{proposition}[Variable-width King moments]
\label{prop:moment}
Let \(l\in\Nzero\), \(j\in\R\), \(k>0\), and
\(\rho,\varsigma>0\).
The absolute moment of \(\mathcal K_l^{[\rho]}\) in
\eqref{eq:radial-moment} is finite if and only if \(j>-l-3\).
For such \(j\), put
\(\xi=k/2\) and \(s_{j,l}=(j+l+3)/2\).  Then
\begin{align}
 \mu_j[\mathcal K_l^{[\rho]}(\cdot;k,\varsigma)]
 &=(2l+1)\Gamma(s_{j,l})(\rho\varsigma)^j\xi^l
 \notag\\[-2pt]
 &\qquad{}\times
 {}_1\widetilde F_1\!
   \left(\frac{l-j}{2};\,l+\frac32;\,-\xi^2\right),
 \label{eq:moment-Klm}
\end{align}
\noindent where
\({}_1\widetilde F_1(a;b;z)={}_1F_1(a;b;z)/\Gamma(b)\).
In particular, with the dimensionless shift \(k\) held fixed,
\begin{equation}
 \mu_j[\mathcal K_l^{[\rho]}(\cdot;k,\varsigma)]
 =\rho^j\,\mu_j[\mathcal K_l(\cdot;k,\varsigma)].
 \label{eq:moment-width-scaling}
\end{equation}
\end{proposition}

\begin{proof}
At the origin the kernel is a positive constant times
\(\vh^l(1+O(\vh^2))\).  At infinity it is bounded by a
constant times \(e^{-\vh^2/\rho^2+k\vh/\rho}\)
\cite[Eqs.~10.52.1 and~10.52.5]{NIST2010}.
Thus absolute moment existence is equivalent to \(j+l+2>-1\).
The scaled definition gives
\[
 \mathcal K_l^{[\rho]}(\vh;k,\varsigma)
 =\rho^{-3}\mathcal K_l(\vh/\rho;k,\varsigma).
\]
\noindent The substitution \(\vh=\rho x\) in \eqref{eq:radial-moment}
proves \eqref{eq:moment-width-scaling}.

For the fixed-width moment, insert the positive-coefficient
Bessel series \eqref{eq:anl}.  Tonelli's theorem permits termwise
integration, and
\(\int_0^\infty x^{j+l+2+2n}e^{-x^2}\dd x
=\tfrac12\Gamma(s_{j,l}+n)\).
Summing the resulting hypergeometric series gives, for \(j>-l-3\),

\[
 \int_0^\infty x^{j+2}e^{-x^2}i_l(kx)\dd x
 =\frac{\sqrt\pi}{4}\xi^l\Gamma(s_{j,l})
   {}_1\widetilde F_1(s_{j,l};\,l+\tfrac32;\,\xi^2).
\]
\noindent Multiplication by \(4\pi\varsigma^{j+3}A_l(k,\varsigma)\)
and Kummer's transformation
\[
 e^{-z}{}_1\widetilde F_1(a;b;z)
 ={}_1\widetilde F_1(b-a;b;-z)
\]
\noindent \cite[Eq.~13.2.39]{NIST2010} yield the \(\rho=1\) case of
\eqref{eq:moment-Klm}.  The scaling identity completes the proof.
\end{proof}

For \(k=0,l=0\), the centered Gaussian has, for \(j>-3\),
\[
 \mu_j[\mathcal K_0^{[\rho]}(\cdot;0,\varsigma)]
 =(\rho\varsigma)^j
    \frac{\Gamma((j+3)/2)}{\Gamma(3/2)}.
\]
\noindent For \(k=0,l>0\), the kernel is identically zero and every absolute
moment is zero; this case is treated directly, not by substituting
into \eqref{eq:moment-Klm} at a pole of the gamma factor.
In particular, every \(l=0\) kernel, including \(k=0\), satisfies
\(\mu_0[\mathcal K_0^{[\rho]}]=1\).
For a component at a fixed physical center, changing its width also
changes \(k=2|\bm u|/(\rho\varsigma)\); that dependence must be retained
when using \eqref{eq:moment-Klm}.

\subsection{Finite mixtures and normalization}

For the finite modal amplitude \(q_l^m\) in
\eqref{eq:discrete-KMM-definition}, abbreviate
\(c_s=c_{l,s}^{m}\), \(\rho_s=\rho_{l,s}^{m}\),
\(k_s=k_{l,s}^{m}\), and \(\xi_s=k_s/2\).
Under \(j>-l-3\), termwise integration gives its moment
\begin{align}
 \mathcal M_{j,l}^{m}[q]
 =(2l+1)\Gamma(s_{j,l})m_a n_a\varsigma^j
      \sum_{s=1}^{S_l^m}c_s\rho_s^j\xi_s^l
 {}_1\widetilde F_1\!
   \left(\frac{l-j}{2};\,l+\frac32;\,-\xi_s^2\right).
 \label{eq:moment-mixture}
\end{align}
\noindent Zero-shift contributions are interpreted as above, with
\(\xi_s^0=1\) for \(l=0\).
The effective coefficient \(c_s\) includes the angular direction
factor in \eqref{eq:f-Laplace-King}.
The condition \(j>-l-3\) ensures termwise integrability.
It need not be necessary for a signed sum with cancellations.

Setting \(j=l=m=0\) in \eqref{eq:moment-mixture} gives the
monopole mass identity
\begin{equation}
 \mu_0[q_0^0]
 =\frac{\mathcal M_{0,0}^{0}[q]}{m_a n_a}
 =\sum_{s=1}^{S_0^0}c_{0,s}^{0}.
 \label{eq:moment-mixture-M000}
\end{equation}
\noindent This recovers \eqref{eq:finite-KMM-normalization}, independently
of component widths.
In the GMM specialization \eqref{eq:variable-width-isotropic-GMM},
the monopole coefficients equal the probability weights.
For a general modal KMM, the coefficient sum fixes total mass;
the remaining probability constraints are those in
\eqref{eq:LK-probability-constraints}.

For fixed shifts and widths, prescribed moments give linear
constraints on the coefficients.
The integrability conditions above justify termwise evaluation;
Remark~\ref{rem:moment-convergence} gives sufficient conditions
for convergence along approximation sequences.

\bibliography{Mathbib}

\end{document}